\documentclass[journal,10pt,letterpaper,romanappendices]{IEEEtran}

\usepackage{amsmath,amssymb,amsfonts,amsthm}
\usepackage{algorithmic}
\usepackage{algorithm}
\usepackage{array}
\usepackage{textcomp}
\usepackage{stfloats}
\usepackage{url}
\usepackage{verbatim}
\usepackage{graphicx}
\usepackage{subcaption}
\usepackage{cite}
\usepackage{xcolor}
\usepackage{braket}
\usepackage[nolinks=true,bookmarks=false]{hyperref}
\usepackage[T1]{fontenc}
\usepackage{gensymb}
\usepackage{orcidlink}
\usepackage{braket}
\usepackage[utf8]{inputenc}

\def\BibTeX{{\rm B\kern-.05em{\sc i\kern-.025em b}\kern-.08em
    T\kern-.1667em\lower.7ex\hbox{E}\kern-.125emX}}

\newcommand{\nb}{\bar{n}_\mathrm{B}}
\newcommand{\im}{j}
\newcommand{\dif}{\mathop{}\!\mathrm{d}}

\DeclareMathOperator{\tr}{tr}

\theoremstyle{definition}
\newtheorem{definition}{Definition}

\newtheorem{theorem}{Theorem}
\newtheorem{lemma}{Lemma}

\allowdisplaybreaks

\begin{document}

\title{Covert Signaling for Communication and Sensing over the Bosonic Channels}

\author{Tianrui Tan\orcidlink{0009-0005-4473-7586}, Evan J.~D.~Anderson\orcidlink{0000-0002-0894-1695}, Michael S.~Bullock\orcidlink{0000-0002-3528-7473}, and Boulat A.~Bash\orcidlink{0000-0002-1205-3906}
\thanks{Tianrui Tan is with the Dept.~of Electrical and Computer Engineering, University of Arizona, Tucson, AZ, USA (email: tantianrui@arizona.edu)}
\thanks{Evan Anderson is with the Dept.~of Electrical and Computer Engineering, University of Maryland,
College Park, MD, USA
(email: ejdanderson@arizona.edu)}
\thanks{Michael S.~Bullock is with the Manning College of Information and Computer Sciences and the Dept. of Electrical and Computer Engineering, USA, University of Massachusetts, Amherst, MA, USA (email: mbullock@umass.edu)}
\thanks{Boulat Bash is with the Dept.~of Electrical and Computer Engineering, and the Wyant College of Optical Sciences, University of Arizona, Tucson, AZ, USA (email: boulat@arizona.edu)}}

\maketitle

\begin{abstract}
Preventing signal detection in communication and  active sensing requires careful control of transmission power.
In fact, the square-root laws (SRL) for covert classical and quantum communication and sensing prescribe that the average output energy per channel use scales as $1/\sqrt{n}$ for $n$ channel uses.
\emph{Diffuse} and \emph{sparse} signaling achieve this. The former transmits signals whose energy decays as $1/\sqrt{n}$ over all $n$ channel uses, which is convenient for mathematical analysis.
The latter transmits constant-energy signals only approximately $\propto\sqrt{n}$ times out of $n$ channel uses, remaining silent on the others. 
This offers significant practical advantages in compatibility with modern digital transmitters. 
Here, we study sparse signaling over the lossy thermal-noise bosonic channel, which is a quantum model of many practical channels (including optical, microwave, and radio-frequency).
We characterize the input signal state that minimizes  detectability.  We find an unintuitive optimal quantum state structure: a mixture of just two consecutive photon-number states. In particular, in the low-brightness regime, the optimal signal state is a mixture of
vacuum and a single photon.
Since these states are generally suboptimal for both communication and active sensing, we explore the resulting trade-off and identify input-power thresholds for transitions between optimizing covertness and performance in communication and sensing tasks. 
\end{abstract}

\section{Introduction}
\label{sec:introduction}
Covert, or low-probability-of-detection/intercept (LPD/LPI) signaling protects transmissions from adversarial detection. Although spread spectrum has provided covertness since the 1940s \cite{simon94ssh}, investigation of its fundamental limits began fairly recently. Initial information-theory study of covert communication over classical additive white Gaussian noise (AWGN) channels \cite{bash12sqrtlawisit, bash13squarerootjsac} was quickly followed up by work on quantum bosonic channels \cite{bash13quantumlpdisit, bash15covertbosoniccomm, bullock20discretemod, gagatsos20codingcovcomm}. AWGN and bosonic channels are, respectively, classical and quantum descriptions of practical channels such as optical, microwave, and radio-frequency (RF). This progress motivated the study of fundamental limits of classical \cite{goeckel17sensinglinsystems-asilomar} and quantum \cite{bash17qcovertsensingisit, gagatsos19floodlightsensor, tahmasbi20sensing, tahmasbi21signalingcovert} covert active sensing. 

A fundamental feature of covert communication and sensing is the square-root law (SRL), which limits the energy per channel use that transmitter Alice can emit without detection to be $\propto 1/\sqrt{n}$ over a blocklength of $n$ channel uses.  This limits her to transmitting only $\mathcal{O}(\sqrt{n})$ covert bits reliably over that blocklength; attempting to transmit more either makes her transmission detectable or prevents reliable decoding at Bob \cite{bash12sqrtlawisit, bash13squarerootjsac}. The SRL also holds for bosonic channels without entanglement assistance \cite{bullock20discretemod, gagatsos20codingcovcomm}.
In sensing, the SRL manifests through $\mathcal{O}(1/\sqrt{n})$ rather than $\mathcal{O}(1/n)$ scaling in  the estimator mean square error decay \cite{goeckel17sensinglinsystems-asilomar, bash17qcovertsensingisit, gagatsos19floodlightsensor} and $\mathcal{O}(\sqrt{n})$ rather than $\mathcal{O}(n)$ detection error exponent \cite{tahmasbi20sensing, tahmasbi21signalingcovert}.
Although the Shannon and Holevo capacities in bits per channel use are zero for a covert channel, a non-trivial amount of information can still be transmitted. This is fully characterized and adapted to classical discrete memoryless channels \cite{bloch15covert, wang15covert, tahmasbi16covertdmc2ndorder} and their quantum counterparts \cite{azadeh16quantumcovert-isitarxiv, wang16cq-srlconverse,  bullock2025fundamentallimitscovertcommunication}, with much other work following. Furthermore, covert networks \cite{soltani18netlpd, arumugam16broadcast, arumugam18mac, tan18covertbc, azadeh18covertmultihop, cho19covertscaling}, coding strategies \cite{zhang16covertcodes, bloch17ppm, wang21covcodes}, asynchrony \cite{bash16timing, arumugam16async, goeckel16timing}, jammer assistance \cite{lee14posratecovert, sobers17jammer} and many other topics have since been investigated. 
A tutorial \cite{bash15covertcommmag} and a comprehensive survey \cite{chen23covcommssurvey} are available.

We explore the signaling necessary to achieve the SRL.
The energy budget for covert transmission over $n$ channel uses scales as $\mathcal{O}(\sqrt{n})$ for classical AWGN \cite{bash12sqrtlawisit, bash13squarerootjsac} and quantum bosonic  \cite{bash13quantumlpdisit, bash15covertbosoniccomm, bullock20discretemod, gagatsos20codingcovcomm} channels. There are two methods for allocating this energy: \emph{diffuse signaling} spreads it over all $n$ channel uses. Thus, energy per channel use decays $\propto 1/\sqrt{n}$. While this is mathematically convenient, e.g., in enabling Taylor series analysis, it is impractical, especially since analog-to-digital and digital-to-analog converters (ADCs/DACs) are limited by finite resolution. This forces the use of an alternative method: \emph{sparse signaling}, which allocates a fixed amount of energy to each of $\propto\sqrt{n}$ randomly selected channel uses. 
The central question addressed here is: what is the structure of the optimal quantum state for achieving covertness via sparse signaling, and what impact does it have on the performance of communication and sensing systems that use it?

Indeed, this paper is motivated by recent experiments in RF that must employ sparse signaling \cite{bali2025experimental-milcom, bali2026experimental}, as well as previous studies of covert signaling for active sensing \cite{tahmasbi20sensing, tahmasbi21signalingcovert}.
Surprisingly, we find that the optimal input state for sparse signaling is a mixture of at most two consecutive Fock (photon-number) states. In the low-brightness regime, this reduces to a vacuum single-photon mixture and attains the global optimum by matching the bound in \cite[Th.~1]{bullock20discretemod}. This presents a trade-off between optimizing the input for covertness vs.~performance in communication and sensing tasks, which we explore.

The rest of the paper is organized as follows.  Section~\ref{sec:preliminaries} introduces the
bosonic channel model, the communication and sensing settings, the covertness criterion, and the sparse-signaling model. Section~\ref{sec:primary-results} derives the optimal sparse signaling state structure.  Section~\ref{sec:covert_regime} explores its communication and sensing performance, and Section~\ref{sec:conclusion} concludes.

\section{Preliminaries}
\label{sec:preliminaries}

\begin{figure}[htb]
\centering
\includegraphics[width=0.33\columnwidth]{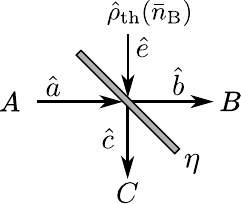}
\caption{The lossy thermal-noise bosonic channel $\mathcal{E}_{A\to BC}^{\eta,\nb}$ is characterized by transmissivity $\eta$ and thermal photon number $\bar{n}_{\text{B}}$, with input subsystem $A$ and output subsystems $B$ and $C$. Environment is in a thermal state $\hat{\rho}_{\rm th}(\nb)$. The modal annihilation operators $\hat{a}$, $\hat{b}$, $\hat{c}$ and $\hat{e}$ describe the channel input-output relationships $\hat{b}=\sqrt{\eta}\hat{a}+\sqrt{1-\eta}\hat{e}$ and $\hat{c}=\sqrt{1-\eta}\hat{a}-\sqrt{\eta}\hat{e}$. 
\label{fig:channel-model}}
\end{figure}

\subsection{System Model}
\label{subsec:systemmodel} 
We briefly review the necessary bosonic systems and quantum optics background, deferring the details to a survey \cite{weedbrook12gaussianQIrmp} and textbooks \cite{scully97quantumoptics, orszag16quantumotpics}.
We employ a lossy thermal-noise bosonic channel modeled by a beam splitter with the environment
mean photon number $\bar{n}_\text{B}$ and transmittance $\eta$ denoted by a linear, completely positive, and trace-preserving (CPTP) map $\mathcal{E}_{A\to BC}^{\eta,\bar{n}_\text{B}}$. Fig.~\ref{fig:channel-model} depicts it. An electromagnetic field mode at the transmitter's center frequency is the fundamental unit of information transmission in our work, akin to channel use in classical information theory.  System $A$ describes the input mode, while system $E$ is the input from the environment in a thermal state $\hat{\rho}_{\text{th}}\left(\bar{n}_\text{B}\right)$. Systems $B$ and $C$ describe the outputs. One of the output modes is to the intended receiver, while the other is signal loss  to the environment and/or the adversary. 

For a single bosonic mode, the Fock (or photon-number) state  $\ket{i}$ describes the state containing exactly $i$ photons. It is an eigenbasis of the photon-number operator $\hat N=\hat a^\dagger \hat a$, i.e., $\hat N \ket{i}= i \ket{i}$,
where $\hat a$ and $\hat a^\dagger$ are the annihilation and creation operators. The set $\{\ket{i}\}_{i\ge 0}$ forms an orthonormal basis for the Hilbert space of the single bosonic mode, with
$\sum_{i=0}^\infty \ket{i}\bra{i}=\mathbb{I}$ \cite{weedbrook12gaussianQIrmp, scully97quantumoptics, orszag16quantumotpics}
Any quantum state of a bosonic mode can therefore be represented in the Fock basis: $
\hat{\rho}=\sum_{i,i^\prime=0}^{\infty}\rho_{ii^\prime}\ket{i}\bra{i^\prime}$, where the diagonal elements $\rho_{ii}$ specify the photon-number distribution, and the off-diagonal elements $\rho_{i,i^\prime\neq i}$ encode phase relations.
For example, the thermal state describing the environment of $\mathcal{E}_{A\to BC}^{\eta,\bar{n}_\text{B}}$ is $\hat{\rho}_{\text{th}}\left(\bar{n}_\text{B}\right)=\sum_{i=0}^\infty\lambda_i\ket{i}\bra{i}$, with $\lambda_i=\frac{(\bar{n}_B)^i}{(1+\bar{n}_B)^{i+1}}$. A vacuum state $\ket{0}\bra{0}=\hat{\rho}_{\text{th}}\left(0\right)$   \cite{weedbrook12gaussianQIrmp, scully97quantumoptics, orszag16quantumotpics}.
We use $n$ orthogonal modes, i.e., Alice transmits an $n$-mode  state $\hat{\rho}_A^{n}$ in both covert communication and active sensing.

\begin{figure*}[t]
    \centering
    \begin{subfigure}[t]{0.52\textwidth}
        \centering
        \includegraphics[width=\linewidth]{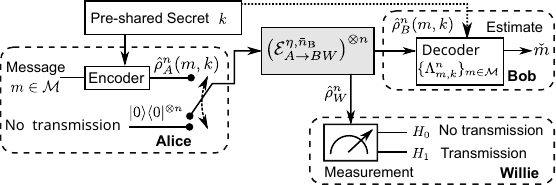}
        \caption{Covert communication}
        \label{fig:system-comms}
    \end{subfigure}
    \hfill
    \begin{subfigure}[t]{0.42\textwidth}
        \centering
        \includegraphics[width=\linewidth]{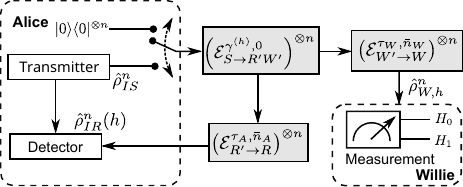}
        \caption{Covert sensing}
        \label{fig:system-sensing}
    \end{subfigure}
    \caption{Systems models for covert communication and sensing described in Section \ref{subsec:systemmodel}. In both models, the adversary, Willie, must decide whether Alice transmits based on his observed state. When she does not transmit, he receives the thermal state.}
    \label{fig:two-column-wide}
\end{figure*}
\subsubsection{Communication Model}
\label{subsubsec:commsmodel}

As in our earlier works \cite{bash13quantumlpdisit, bash15covertbosoniccomm, bullock20discretemod, gagatsos20codingcovcomm}, we model communication by a single lossy thermal-noise bosonic channel $\mathcal{E}^{\eta,\bar{n}_{\text{B}}}_{A\to BW}$ depicted in Fig.~\ref{fig:system-comms}. 
Alice and Bob share a random classical secret key $k\in\mathcal{K}$  that is inaccessible to the adversary Willie; we defer the characterization of its size to future work. 
Alice encodes an equiprobable message $m\in\mathcal{M}$  in an $n$-mode state $\hat{\rho}_{A}^{n}(m,k)$. An optimal random code outputs a tensor product of $n$ independent and identically distributed (i.i.d.) states, thus $\hat{\rho}_{A}^{n}(m,k)=\bigotimes_{i=1}^n\hat{\rho}_{A,i}(m,k)$ \cite{giovannetti_ultimate_2014}.
 These are transmitted to Bob via $n$ independent uses of $\mathcal{E}^{\eta,\bar{n}_{\text{B}}}_{A\to BW}$, who receives the $n$-mode output state $\hat{\rho}_{B}^{n}(m,k)=\tr_{W}\left[\left(\mathcal{E}_{A\to BW}^{\eta,\bar n_B}\right)^{\otimes n}\left(
\hat{\rho}_{A}^{n}(m,k)\right)\right]$, and uses $k$ to 
estimate $m$ with vanishing error probability \cite{gagatsos20codingcovcomm}. 

\subsubsection{Sensing Model}
\label{subsubsec:sensingmodel}

We study the covert version  \cite{tahmasbi20sensing, tahmasbi21signalingcovert} of quantum-illumination-based binary target detection  \cite{tan08qigaussianstates}, depicted in Fig.~\ref{fig:system-sensing}. 
Alice prepares an $n$-mode signal-reference state $\hat{\rho}_{IS}^{n}$ and retains the reference (``idler'') modes $I$.
As in \cite{tahmasbi20sensing, tahmasbi21signalingcovert, tan08qigaussianstates}, we employ product states: $\hat{\rho}_{IS}^{n}=\hat{\rho}_{IS}^{\otimes n}$.
Signal modes $S$ interrogate the target, modeled by a pure-loss bosonic channel $\mathcal{E}^{\gamma^{(h)},0}_{S\to R^\prime W^\prime}$, where transmittance to Willie $\gamma^{(h)}$ depends on whether the target is absent ($h=0$,  $\gamma^{(0)}=1$) or present ($h=1$, $\gamma^{(1)}<1$).
Lossy thermal-noise bosonic channels $\mathcal{E}^{\tau_A,\bar{n}_A}_{R^\prime\to R}$ and $\mathcal{E}^{\tau_W,\bar{n}_W}_{W^\prime\to W}$ corrupt the light that the target reflects to Alice (modes $R^\prime$) and transmits to Willie (modes $W^\prime$), respectively.
Alice detects the target by a hypothesis test on the reference and returned modes $I$ and $R$ in the state
$\hat{\rho}_{IR}^{n}(h)=
\left(\mathcal{I}_{I\to I}\otimes\mathcal{E}^{(1-\gamma^{(h)})\tau_A,(1-\tau_A)\bar n_A}_{S\to R}\right)^{\otimes n}\left(\hat{\rho}_{IS}^n\right)$,
where $\mathcal{I}_{I\to I}$ is an identity channel and 
$
    \mathcal{E}^{(1-\gamma^{(h)})\tau_A,(1-\tau_A)\bar n_A}_{S\to R}=\mathcal{E}^{\tau_A,\bar n_A}_{R'\to R}
    \left(
        \tr_{W'}
        \left[
            \mathcal{E}^{\gamma^{(h)},0}_{S\to R'W'}(\hat\rho_S)
        \right]
    \right)$.
Since probes are not returned when target is absent, $\hat{\rho}_{IR}^{n}(0)=\tr_{S}\left[\hat{\rho}_{I^nS^n}\right]
\otimes
\hat{\rho}_{\mathrm{th}}\left((1-\tau_A)\bar n_A\right)^{\otimes n}$.
For equal prior probabilities $\Pr(H_0 \text{ true})=\Pr(H_1 \text{ true})=1/2$  of hypotheses $H_0$ and $H_1$ that the target is absent and present, Alice's detection error
probability $P_{ e}^{(s)}
=\frac{P_{\text{FA}}+P_{\text{MD}}}{2}$, where $P_{\text{FA}}=
\Pr\left(\text{choose } H_1 \mid H_0 \text{ true}\right)$, and $P_{\text{MD}}
=
\Pr\left(\text{choose } H_0 \mid H_1 \text{ true}\right)$.
Helstrom limit \cite{helstrom76quantumdetect}~\cite[Sec.~9.1.4]{wilde16quantumit2ed} yields
\begin{align}
P_{\mathrm e}^{(\mathrm s)}
\geq
\frac{1}{2}
\left(
1-\frac{1}{2}
\left\|
\hat{\rho}_{IR}^{n}(1)
-
\hat{\rho}_{IR}^{n}(0)
\right\|_1
\right),
\label{eq:helstrom_limit_sensing}
\end{align}
where $\left\|\hat A\right\|_1=\tr\left[\sqrt{\hat A^\dagger \hat A}\right]$ is the trace norm.

In both communication and active sensing, Alice’s objective is to transmit while remaining undetected by the adversary Willie. This motivates the covertness criteria introduced next.

\subsection{Covertness criteria}
\label{subsec:covertness}

Covertness quantifies Willie’s ability to distinguish whether Alice is transmitting or not.
When silent, she inputs innocent vacuum state $\hat{\rho}^n_{A}=\ket{0}\bra{0}^{\otimes n}$,
which yields an $n$-mode thermal state at Willie,
$\hat{\rho}_{\text{th}}^{\otimes n}(\eta \bar n_B)$ for communication and $\hat{\rho}_{\text{th}}^{\otimes n}\left(\left(1-\tau_W\right)\bar n_W\right)$ for sensing. 
When Alice transmits in the communication setting,
the $n$-mode state $\hat{\rho}_{W}^{n}=\frac{1}{|\mathcal{M}|}\frac{1}{|\mathcal{K}|}\sum_{m,k}\tr_B\left[\left[\mathcal{E}^{\eta,\bar{n}_{\text{B}}}_{A\to BW}\right]^{\otimes n}\left(\hat{\rho}_{A}^{n}(m,k)\right)\right]$ is Willie's output of channel $\mathcal{E}^{\eta,\bar{n}_{\text{B}}}_{A\to BW}$ averaged over messages in $\mathcal{M}$ and the secret  from $\mathcal{K}$ inaccessible by Willie. 
As in \cite{bullock20discretemod, gagatsos20codingcovcomm}, we assume $\mathcal{K}$ is a random code, implying that $\hat{\rho}_{W}^{n}=\hat{\rho}_{W}^{\otimes n}$.
In the sensing setting, Willie’s received state depends on the target: $\hat{\rho}_{W,h}^{n}=\left(\mathcal{E}^{\tau_W,\bar n_W}_{W'\to W}\right)^{\otimes n}
\left(
\tr_{IR^\prime}
\left[
\left(
\mathcal{I}_{I\to I}
\otimes
\mathcal{E}^{\gamma^{(h)},0}_{S\to R^\prime W^\prime}
\right)^{\otimes n}
\left(\hat{\rho}_{IS}^{n}\right)
\right]
\right)$, where $h\in\{0,1\}$ indicates target's presence.
Since here $\hat{\rho}_{IS}^{n}=\hat{\rho}_{IS}^{\otimes n}$, $\hat{\rho}_{W,h}^{n}=\hat{\rho}_{W,h}^{\otimes n}$.

As in \cite{bash12sqrtlawisit, bash13squarerootjsac, bash13quantumlpdisit, bash15covertbosoniccomm, bullock20discretemod, gagatsos20codingcovcomm, goeckel17sensinglinsystems-asilomar, bash17qcovertsensingisit, gagatsos19floodlightsensor, tahmasbi20sensing, tahmasbi21signalingcovert}, we call a system covert if Willie’s average detection error probability $P_{\rm e}^{\text{(w)}}\geq\frac{1}{2}-\delta$ for small $\delta>0$. Like $P_{\rm e}^{\text{(s)}}$ in \eqref{eq:helstrom_limit_sensing}, for communication it is:
\begin{align}
    P_{\rm e}^{\text{(w)}}&\geq \frac{1}{2}
\left(
1
-
\frac{1}{2}
\left\|
\hat{\rho}_{W}^{\otimes n}
-
\hat{\rho}_{\mathrm{th}}^{\otimes n}(\eta \bar n_B)
\right\|_1
\right),\label{eq:pe_willie_comms}
\end{align}
whereas for sensing it is: 
\begin{align}
    P_{\rm e}^{\text{(w)}}&\geq \frac{1}{2}
\left(
1
-
\max_{h\in\{0,1\}}\frac{1}{2}
\left\|
\hat{\rho}_{W,h}^{\otimes n}
-
\hat{\rho}_{\mathrm{th}}^{\otimes n}\left(\left(1-\tau_W\right)\bar n_W\right)
\right\|_1
\right).\label{eq:pe_willie_sensing}
\end{align}
Numerics suggest that $h=0$ minimizes the lower bound in \eqref{eq:pe_willie_sensing}; we defer analytical confirmation to future work.
The trace distance is mathematically
inconvenient; however, quantum Pinsker’s inequality \cite[Th.~11.9.1]{wilde16quantumit2ed}, $\frac{1}{2}
\left\|
\hat{\rho}
-
\hat{\sigma}
\right\|_1
\leq
\sqrt{
\frac{1}{2}
D\left(
\hat{\rho}
\middle\|
\hat{\sigma}
\right)
}$, upper bounds it via the quantum relative entropy (QRE) $D(\hat{\rho}\|\hat{\sigma})=\tr[\hat{\rho}\log\hat{\rho}-\hat{\rho}\log\hat{\sigma}]$. 
Thus, as in \cite{bash12sqrtlawisit, bash13squarerootjsac, bash13quantumlpdisit, bash15covertbosoniccomm, bullock20discretemod, gagatsos20codingcovcomm, goeckel17sensinglinsystems-asilomar, bash17qcovertsensingisit, gagatsos19floodlightsensor, tahmasbi20sensing, tahmasbi21signalingcovert}, we use it to define our covertness constraint.
Since QRE is additive for product states, the communication and sensing systems are covert if:
\begin{align}
    nD\left(\hat{\rho}_{W} \middle\|
\hat{\rho}_{\mathrm{th}}(\eta \bar n_B)\right)&\leq \delta_{\text{QRE}}\label{eq:cov_req_comms}\\
    \max_{h\in\{0,1\}}nD\left(\hat{\rho}_{W,h}\middle\| \hat{\rho}_{\mathrm{th}}\left(\left(1-\tau_W\right)\bar n_W\right)\right)&\leq \delta_{\text{QRE}}.\label{eq:cov_req_sensing}
\end{align}
Setting $\delta_{\text{QRE}}=8\delta^2$ meets the constraints in  \eqref{eq:pe_willie_comms} and \eqref{eq:pe_willie_sensing}.

\subsection{Sparse Signaling}\label{subsec: Signaling Model}

We focus on the communication covertness constraint in \eqref{eq:cov_req_comms} since our results hold for \eqref{eq:cov_req_sensing} as well.
Per Section \ref{sec:introduction}, we focus on sparse signaling.
Suppose Alice and Bob flip a biased coin $n$ times, selecting the modes corresponding to ``head'' outcomes for transmission.  For ``tail'' outcomes, Alice remains quiet.
If $P(\text{``heads''})=\alpha$, and the outcomes along with the random code are hidden from Willie, he receives
\begin{align}
\label{eq:rho_W}
\hat{\rho}_W = (1-\alpha)\hat{\rho}_\text{th}\left(\eta\bar{n}_\text{B}\right)+ \alpha \hat{\sigma}_W,
\end{align}
where the thermal state
$\hat\rho_\text{th}\left(\eta\bar{n}_\text{B}\right)=\sum_{\ell=0}^{\infty}\lambda_\ell \ket{\ell}\bra{\ell}$,
$\lambda_\ell=\frac{(\eta \bar n_B)^\ell}{(1+\eta \bar n_B)^{\ell+1}}$, is induced by Alice inputting vacuum $\ket{0}\bra{0}$, and $\hat{\sigma}_W =
\tr_{B}
\left[
\mathcal E_{A\to BW}^{\eta,\bar n_B}(\hat\rho_A)
\right]=\mathcal E_{A\to W}^{\eta,\bar n_B}(\hat\rho_A)$,
is the output when Alice inputs $\hat{\rho}_A$, which is an average state over the random code and messages. We assume that $\hat{\rho}_A$'s mean photon number $\tr\left[\hat{N}\hat{\rho}_A\right]=\bar{n}_S$. Allowing only $\alpha$ to decay with $n$, we characterize $\hat{\rho}_A$'s optimal structure next.

\section{Characterization Of  Sparse Signaling}
\label{sec:primary-results}

We first state our main result:
\begin{theorem}[Optimal sparse input state structure]
\label{thm:adjacent_support} For a sparse signaling scheme,
\begin{align}
\label{eq:D_expansion}
D(\hat{\rho}_W\|\hat{\rho}_{\text{th}}\left(\eta\bar{n}_\text{B}\right))
=
\frac{\alpha^2}{2} C_W + o(\alpha^2).
\end{align}
Also, for $\tr[\hat\rho_A\hat N]=\bar n_S$ and $\kappa=\lfloor \bar n_S \rfloor$, $C_W$ is minimized by  
$
\hat{\rho}_{A}
=
(\kappa+1-\bar n_S)\ket{\kappa}\bra{\kappa}
+(\bar n_S-\kappa)\ket{\kappa+1}\bra{\kappa+1}$.
\end{theorem}

Thus, optimal $\hat{\rho}_A$ is a Fock-diagonal state supported on at most two consecutive Fock states.
The proof is developed through the sequence of results below. We first show using data processing inequality and dephasing channel that optimizing $D(\hat{\rho}_W\|\hat{\rho}_{\text{th}}\left(\eta\bar{n}_\text{B}\right))$ can be restricted to Fock-diagonal input states of the form $\hat{\rho}_{A}= \sum_{i\ge 0}\rho_i \ket{i}\bra{i}$. Then, we derive a constant coefficient $C_W$ when $\alpha$ decays to zero with $n$. Hence, the sparse-covertness problem reduces to minimizing $C_W$ over the input distribution $\{\rho_i\}$. Next, we show that the structure of $C_W$ depends on the binomial expectation over $\{\rho_i\}$. Finally, combining this representation with a discrete-convexity argument yields an optimizer supported on at most two consecutive Fock states.

\subsection{Fock-Diagonal Input State is Optimal for Sparse Signaling}
\label{subsec:dephase_reduc}

\begin{definition}
\label{def:phase_U}
(Phase rotation unitary).
$\hat{U}_\phi = e^{-\im\phi \hat N}, \qquad \phi\in[0,2\pi)$, where $\hat N=a^\dagger a$.
\end{definition}

\begin{definition}[Dephasing channel]
\label{def:dephase}
$\Delta(\hat\rho)
=
\sum_{\ell=0}^{\infty}
\ket{\ell}\bra{\ell}\hat{\rho}\ket{\ell}\bra{\ell}$.
Equivalently,
$\Delta(\hat\rho)
    =
    \int_0^{2\pi}
    \hat{U}_\phi\hat\rho \hat{U}_\phi
    \frac{\dif\phi}{2\pi}$.
\end{definition}
The dephasing channel yields a Fock-diagonal state while preserving the mean photon number. Since $\hat{\rho}_\text{th}\left(\eta\bar{n}_\text{B}\right)$ is diagonal in the Fock basis, $\Delta(\hat{\rho}_\text{th}\left(\eta\bar{n}_\text{B}\right))=\hat{\rho}_\text{th}\left(\eta\bar{n}_\text{B}\right)$. 

\begin{definition}(Phase covariant channel) 
\label{def:phase_cov_channel}
A channel described by CPTP map $\mathcal{N}_{A\to B}$ is called phase-covariant if, $\mathcal{N}_{A\to B}(U_\phi \hat\rho U_\phi^\dagger)
=
U_\phi\mathcal{N}_{A\to B}(\hat{\rho})U_\phi^\dagger
\quad \forall \phi$.
\end{definition}
The lossy thermal-noise and dephasing channels are phase-covariant \cite{weedbrook12gaussianQIrmp, bartlett2007reference}.
Thus,
\begin{align}
\label{eq:dephase_alice_dephase willie} 
\Delta(\hat{\rho}_W)
&= \notag 
(1-\alpha)\Delta\left(\hat{\rho}_\text{th}(\eta\bar{n}_\text{B})\right)
+
\alpha\Delta(\hat{\sigma}_W)   \notag\\  
&=(1-\alpha)\hat{\rho}_\text{th}(\eta\bar{n}_\text{B})
+
\alpha\mathcal{E}_{A\to W}^{\eta,\bar{n}_\text{B}}\left(\Delta(\hat\rho_A)\right).
\end{align}
We can restrict our optimization to Fock-diagonal $\hat{\rho}_A$, since by the data-processing inequality for the QRE \cite[Ch.~10.7.2]{wilde16quantumit2ed}, 
\begin{align}
\label{eq:DPI_sparse_applied}D\left(\hat{\rho}_W\|\hat{\rho}_\text{th}\left(\eta\bar{n}_\text{B}\right)\right)
\ge
D\left(\Delta\left(\hat{\rho}_W\right)\|\Delta\left(\hat{\rho}_\text{th}\left(\eta\bar{n}_\text{B}\right)\right)\right).
\end{align}

\subsection{Quantum Relative Entropy Expansion}

Consider diagonal inputs supported on a truncated Fock subspace
$\mathcal{H}_{r} = \mathrm{span}\{\ket{0},\ket{1},\ldots,\ket{r}\}$:
\begin{align}
    \hat{\rho}_A=\sum_{i=0}^{r} \rho_i \ket{i}\bra{i},\quad \rho_i \ge 0,\quad \sum_{i=0}^{r} \rho_i = 1.\label{eq:truncated_input}
\end{align}
Output of the lossy thermal-noise bosonic channel for Fock states and
the associated Laguerre-polynomial phase-space representation are well-known
\cite{weedbrook12gaussianQIrmp, guha04mastersthesis, schleich2001quantum}.
We use these to characterize Willie’s
 output state $\hat{\sigma}_W$ for input in \eqref{eq:truncated_input}. 

\begin{lemma}[Willie’s  output for truncated input]
\label{lemma:willie_output}
For input $\hat{\rho}_A$ in \eqref{eq:truncated_input}, output $\hat{\sigma}_W$ is Fock-diagonal:
$\hat{\sigma}_W=\sum_{\ell=0}^{\infty} p_\ell \ket{\ell}\bra{\ell}$ for
\begin{align}
p_\ell
=
\sum_{i=0}^{r} \rho_i
\int_{0}^{\infty}
e^{-(1+\eta \bar{n}_B)x}
L_\ell(x)
L_i\left((1-\eta)x\right)\dif x ,
\label{eq:pl_integral_lemma}
\end{align}
with $L_n(\cdot)$ the degree-$n$ Laguerre polynomial.
Equivalently:
$p_\ell
=
\sum_{i=0}^{r} \rho_i
\sum_{t=0}^{i}
\sum_{u=0}^{\ell}\sigma_{t,u}(x)$,
where $\sigma_{t,u}(\eta\bar{n}_{\text{B}})=(-1)^{t+u}
\binom{i}{t}
\binom{\ell}{u}
(1-\eta)^t
\frac{(t+u)!}{t!u!}
\frac{1}{(1+x)^{t+u+1}}
$.
\end{lemma}

The proof is in Appendix \ref{ap:Willie_output}. 
Diagonality of $\hat{\sigma}_W$ leads to:
\begin{lemma}
[Sparse-signaling QRE Expansion]
\label{lemma:sparse_qre_expansion}
For input $\hat{\rho}_A$ in \eqref{eq:truncated_input}, \eqref{eq:D_expansion} holds for $\alpha\to0$ and constant
$C_W=\sum_{\ell=0}^{\infty}\frac{(p_\ell-\lambda_\ell)^2}{\lambda_\ell}$.
\end{lemma}

\begin{IEEEproof}[Proof idea]
Diagonality of $\hat{\rho}_W$ and $\hat{\rho}_\text{th}\left(\eta\bar{n}_\text{B}\right)$ simplifies the QRE to the classical relative entropy and reduces the proof to a scalar second-order Taylor expansion in Appendix \ref{ap:proof_QRE_expansion}. 
\end{IEEEproof}

\subsection{Orthogonal Structure Optimization}
The representation of  $C_W$ in Lemma~\ref{lemma:sparse_qre_expansion} is explicit, but unsuitable for structural analysis and optimization of $\hat{\rho}_A$. Indeed, each $p_\ell$ in the expression for $C_W$ in Lemma~\ref{lemma:sparse_qre_expansion} depends on the entire input distribution $\{\rho_i\}$, making it difficult to identify the input state features that increase detectability by Willie. Thus, we derive an orthogonal representation of $C_W$ as a weighted sum of squares of input-dependent moment quantities, which is a key tool for proving Theorem~\ref{thm:adjacent_support}.

    \begin{definition}[Meixner polynomials \protect{\cite[Sec.~9.10]{koekoek2010hypergeometric}}]\label{def:meixner}
For \(0<\zeta<1\) and \(\ell,y\in\{0,1,2,\dots\}\), the Meixner polynomials \(\{M_y(\ell;\zeta)\}_{y\ge 0}\) are defined 
 by the hypergeometric function
$M_y(\ell;\zeta)
=
{}_2F_1\left(
\begin{matrix}
-y,-\ell\\
1
\end{matrix}
;1-\frac{1}{\zeta}
\right)$.
Equivalently, they are characterized by the generating function
$\sum_{y=0}^{\infty} M_y(\ell;\zeta)\nu^y
=
\left(1-\frac{\nu}{\zeta}\right)^{\ell}(1-\nu)^{-\ell-1}$, $|\nu|<1$.
Moreover, they are orthogonal with respect to the geometric distribution
\(\lambda_\ell=(1-\zeta)\zeta^\ell\), namely
$\sum_{\ell=0}^{\infty}\lambda_\ell M_y(\ell;\zeta)M_{y'}(\ell;\zeta)
=
\zeta^{-y}\delta_{y,y'}$.
\end{definition}

\begin{lemma}(Orthogonal form of $C_W$ in Lemma~\ref{lemma:sparse_qre_expansion}).
\label{lemma:dia_perturb_coef}
$
C_W
=
\sum_{y = 1}^r \zeta^{-y} \theta^{2y} \mu_y^2$, 
where
$\mu_y=\sum_{i = y}^r \rho_i \binom{i}{y}$,
$\zeta=\frac{\eta \bar{n}_B}{1+\eta \bar{n}_B}$, and $\theta=\frac{1-\eta}{1+\eta \bar{n}_B}$.
\end{lemma}
\begin{IEEEproof}[Proof idea]
We make the dependence of $C_W$ on $\{\rho_i\}$ explicit: starting from  \eqref{eq:pl_integral_lemma}, we first write $C_W$ as a quadratic form with respect to $\{\lambda_{\ell}\}$. Since $\lambda_\ell$ is geometric, the Meixner polynomials provide a natural orthogonal basis. We expand the normalized coefficients in this basis using the generating functions and express $\frac{p_\ell}{\lambda_\ell}-1$ using the binomial expectation $\mu_y$ of $\{\rho_i\}$.
Finally, substituting this expansion into $C_W$, and invoking Meixner orthogonality, removes the cross terms, yielding the result. The details are in Appendix \ref{ap:proof_orthogonal_C}.
\end{IEEEproof}

Our final ingredient involves discrete convexity.

\begin{definition}
[Discrete convexity \protect{\cite[p.~42, Ex.~1]{niculescu2006convex}}]
\label{def:disc_covex} 
A function $\varphi:\mathbb{Z}_{\ge 0}\to\mathbb{R}$ is discrete convex if, for every $i\ge 0$, the increments $\varphi(i+1)-\varphi(i)$ are nondecreasing in $i$, or
 $\varphi(i+2)-2\varphi(i+1)+\varphi(i)\ge 0$.
\end{definition}

\begin{lemma}(Two-point lower bound)
\label{lemma:two_p_lb}
Let $\varphi$ be discrete convex and let $\mathrm{I}$ be any integer-valued random variable with mean
$\mathbb{E}[\mathrm{I}]=\kappa+\beta,~ \kappa= \lfloor \mathbb{E}[\mathrm{I}]\rfloor,~ \beta\in[0,1).$
Then
$\mathbb{E}[\varphi(\mathrm{I})]\ \ge\ (1-\beta)\varphi(\kappa)+\beta\varphi(\kappa+1)$, 
with equality for $\Pr[\mathrm{I}=\kappa]=1-\beta$ and $\Pr[\mathrm{I}=\kappa+1]=\beta$.    
\end{lemma}
We defer the detailed proof to Appendix \ref{ap:two_p_lb}. 
To use Lemma~\ref{lemma:two_p_lb} in our setting, fix $y\ge 2$, and define
$
\varphi_y(i)=\binom{i}{y}
$
for $i\in\mathbb{Z}_{\ge 0}$. Then, $\varphi_y$ is discrete convex, since two successive applications of Pascal’s identity $\binom{i+1}{y}=\binom{i}{y}+\binom{i}{y-1}$ yield
$\binom{i+2}{y}-2\binom{i+1}{y}+\binom{i}{y}
=
\binom{i}{y-2}
\ge 0$.
Thus, Lemma~\ref{lemma:two_p_lb} applies to binomial expectation $\mu_y$.
We can now prove our main result.
\begin{IEEEproof}[Proof (Theorem~\ref{thm:adjacent_support})]
The dephasing reduction in Section~\ref{subsec:dephase_reduc} allows optimizing only Fock-diagonal inputs with photon-number distributions $\{\rho_i\}$.
For these, Lemma~\ref{lemma:willie_output} and
Lemma~\ref{lemma:sparse_qre_expansion} yield the expansion in \eqref{eq:D_expansion}.
By Lemma~\ref{lemma:dia_perturb_coef}, coefficient $C_W$ is a weighted sum of squared binomial expectations $\mu_y^2$ of the input photon-number distribution $\{\rho_i\}$. 
Since $\mu_1=\sum_{i= 1}^r i\rho_i=\bar n_S$, minimizing $C_W$ reduces to minimizing the terms for $y\geq 2$. 
Let $\bar n_S=\kappa+\beta,\quad \kappa=\lfloor \bar n_S\rfloor,\quad \beta\in[0,1)$.
Applying Lemma~\ref{lemma:two_p_lb} yields
$\mu_y=\sum_{i\ge y}\rho_i\binom{i}{y}\ge(1-\beta)\binom{\kappa}{y}+\beta\binom{\kappa+1}{y}$, 
$y\ge2$.
Equality is attained simultaneously for every  $\mu_y$, $y\ge2$, by the two-point
distribution
    $\rho_\kappa=1-\beta$,
    $\rho_{\kappa+1}=\beta$,
    $\rho_i=0$ for $i\notin\{\kappa,\kappa+1\}$.
Since coefficients $\zeta^{-y}\theta^{2y}\geq0$, it minimizes $C_W$. Setting $r>\bar{n}_S$ completes the proof.
\end{IEEEproof}

\subsection{Optimality in the Low-Brightness Regime}
\label{subsec:low_brightness_sparse_optimum}
Consider two regimes for Alice's transmission:  $0 \le \bar{n}_S \le 1$ and $\bar{n}_S>1$.
We call these low- and high-brightness regimes, respectively.
In the high-brightness regime, the QRE expansion coefficient $C_W$ in \eqref{eq:D_expansion} exceeds the lower bound in \cite[Th.~1]{bullock20discretemod}, rendering sparse signaling suboptimal.
However, the low-brightness regime is, arguably, more operationally significant \cite{pirandola2021limits, sorelli22qitutotial}.
There, the optimal $\hat{\rho}_{A}=(1-\bar{n}_S)\ket{0}\bra{0}+\bar{n}_S \ket{1}\bra{1}$ per Theorem \ref{thm:adjacent_support}.
The corresponding QRE expansion coefficient $C_W=\frac{(1-\eta)^2\bar{n}_S^2}{\eta \bar n_B(1+\eta \bar n_B)}$ in \eqref{eq:D_expansion} matches the lower bound in \cite[Th.~1]{bullock20discretemod}.
Thus, sparse signaling with $\hat{\rho}_{A}$ is optimal in the low-brightness regime, and we next focus on communication and sensing in this setting.

\section{Impact on Covert Communication and Sensing}

\begin{figure*}[htb]
\centering
\subfloat[$\eta=0.4$]{%
    \includegraphics[width=0.31\textwidth]{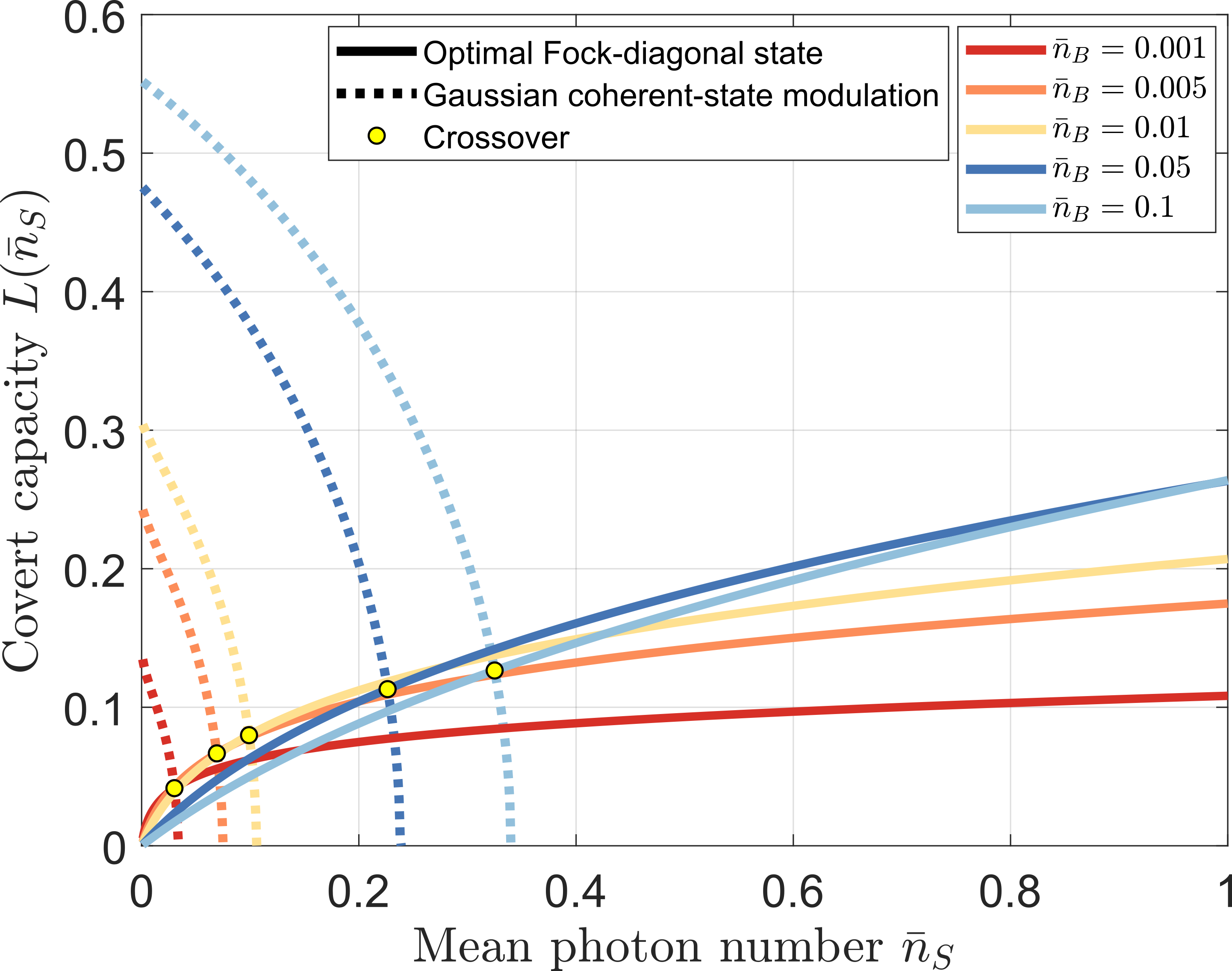}
    \label{fig:comms_comp-1}
}
\hfill
\subfloat[$\eta=0.6$]{%
    \includegraphics[width=0.31\textwidth]{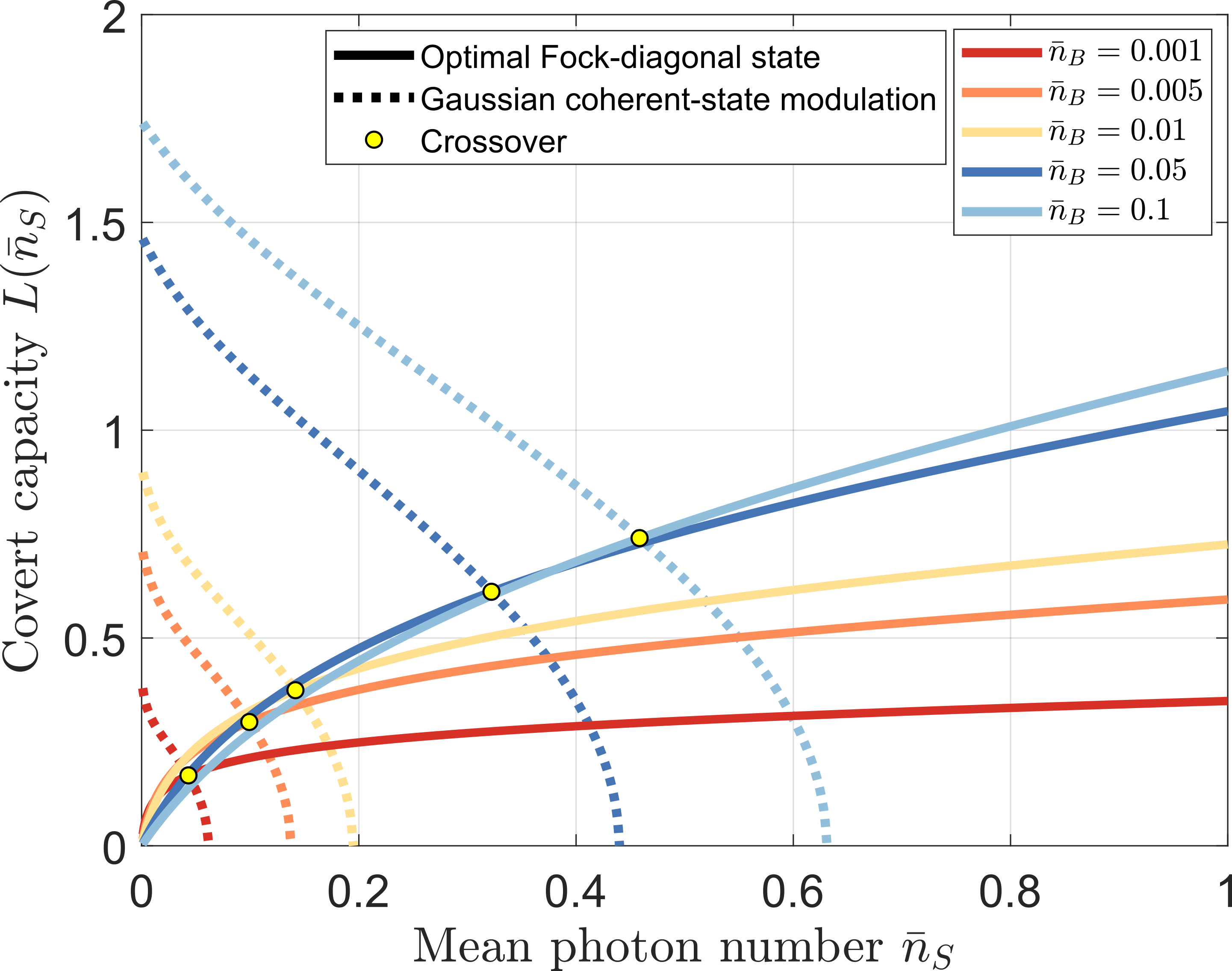}
    \label{fig:comms_comp-2}
}
\hfill
\subfloat[$\eta=0.8$]{%
    \includegraphics[width=0.31\textwidth]{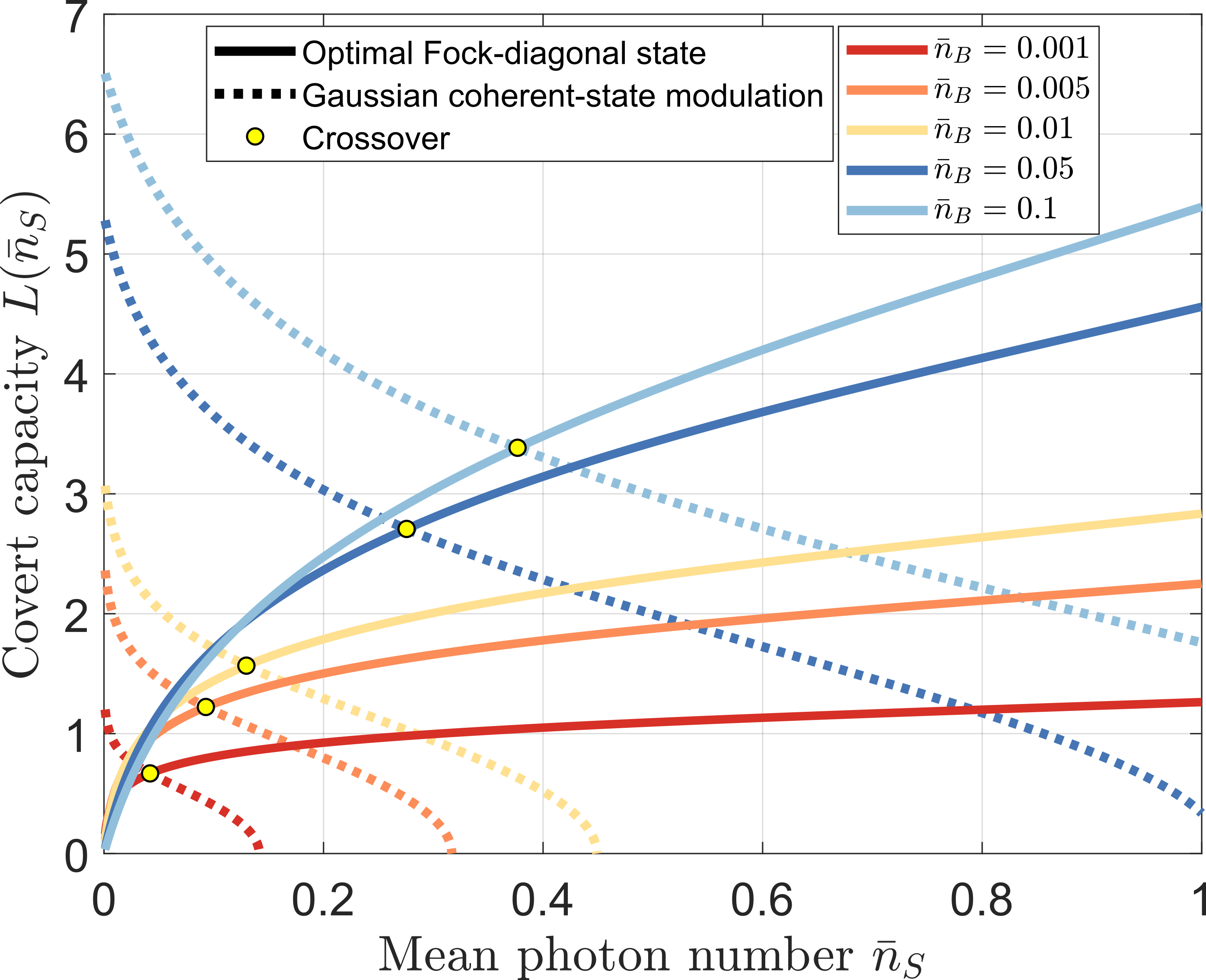}
    \label{fig:comms_comp-3}
}
\caption{Covert communication capacity $L(\bar{n}_S)$ vs.~mean photon number \(\bar n_S\) for sparse signaling over the lossy thermal-noise bosonic channel. The three panels correspond to \(\eta=0.4\), \(\eta=0.6\), and \(\eta=0.8\), respectively. Solid curves show the covertness-optimized Fock-diagonal input; dotted curves show Gaussian coherent-state modulation.
Different colors correspond to different background photon numbers $\bar{n}_B$. Yellow markers denote the performance crossover between the covertness-optimized Fock-diagonal scheme and the Gaussian coherent-state modulation.}
\label{fig:plot_comms_comp}
\end{figure*}

\begin{figure*}[htb]
\centering
\subfloat[$\gamma^{(1)} = 0.4$]{%
    \includegraphics[width=0.31\textwidth]{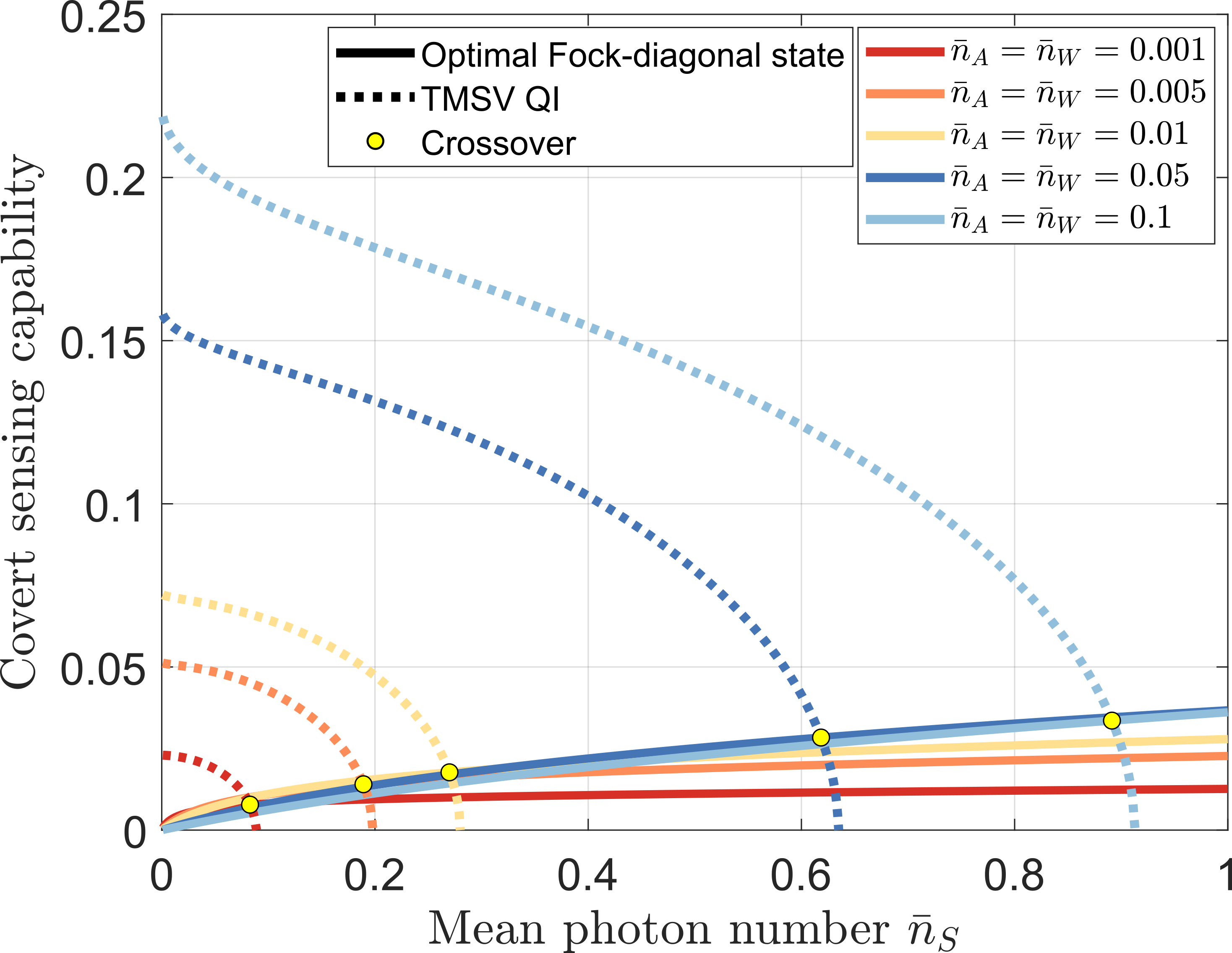}
    \label{fig:sensing_comp-1}
}
\hfill
\subfloat[$\gamma^{(1)} = 0.6$]{%
    \includegraphics[width=0.31\textwidth]{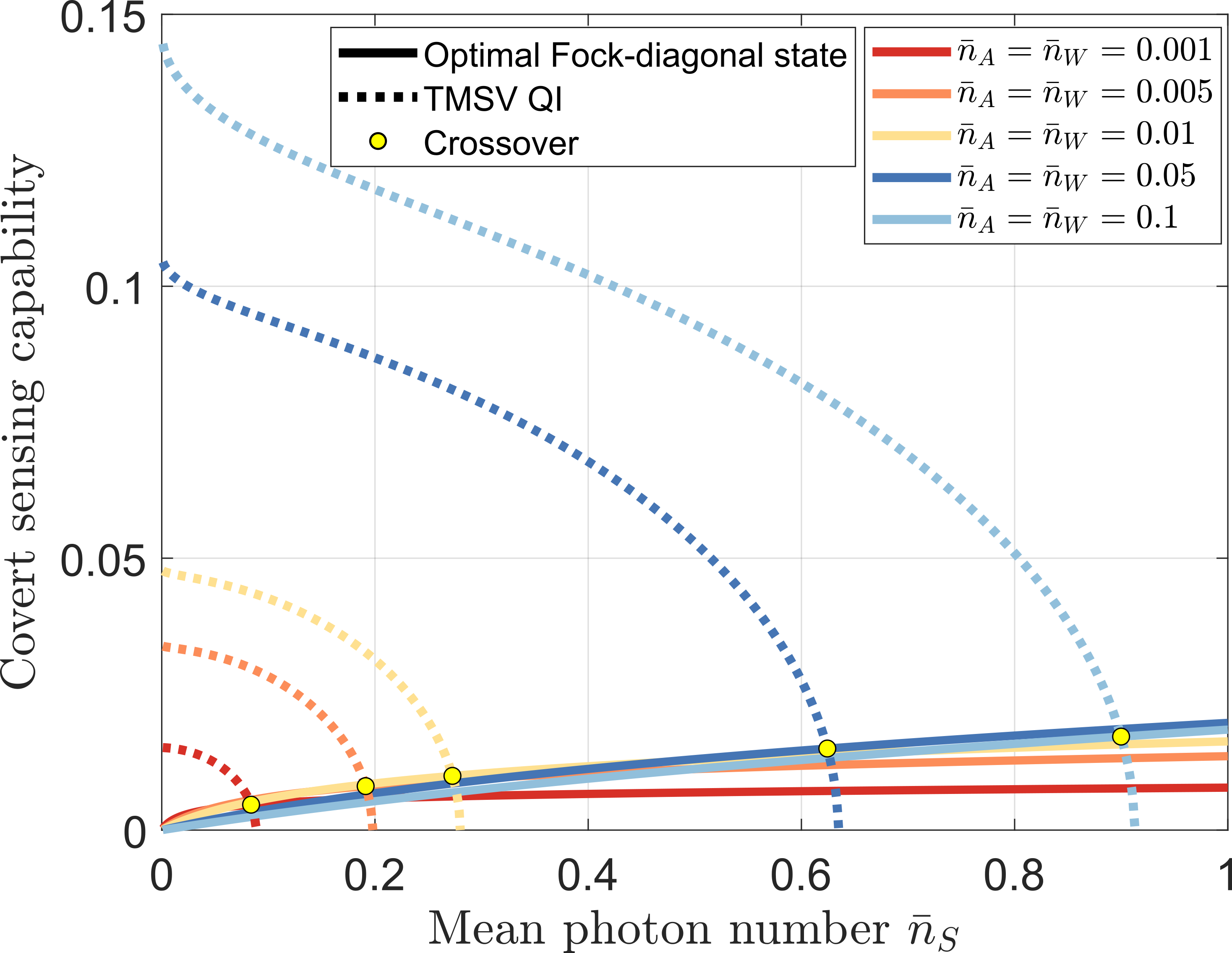}
    \label{fig:sensing_comp-2}
}
\hfill
\subfloat[$\gamma^{(1)} = 0.8$]{%
    \includegraphics[width=0.31\textwidth]{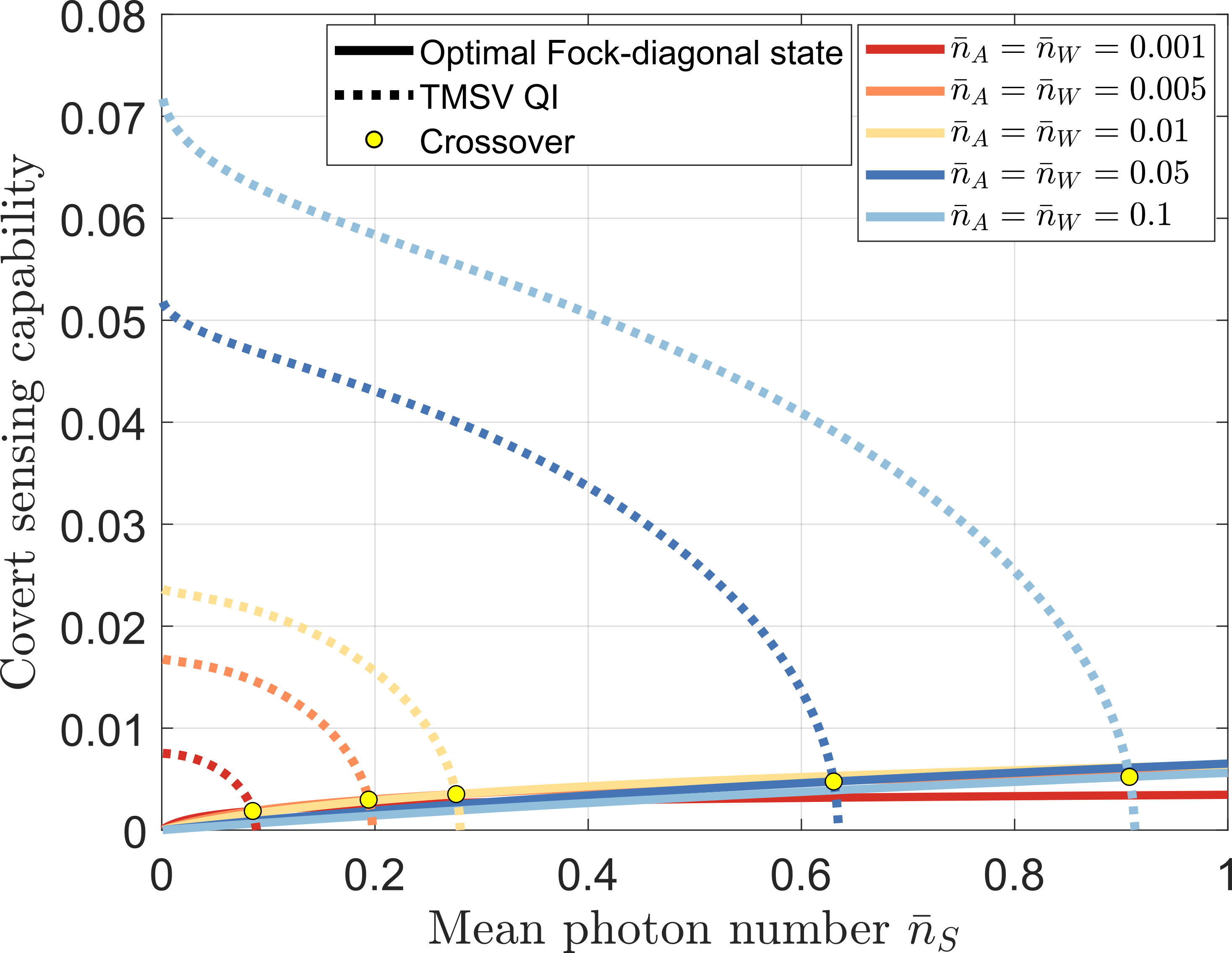}
    \label{fig:sensing_comp-3}
}
\caption{Covert sensing capability vs.~mean photon number \(\bar n_S\) for sparse signaling with  \(\tau_A=\tau_W=0.3\). The three panels correspond to target-present parameters \(\gamma^{(1)}=0.4\), \(\gamma^{(1)}=0.6\), and \(\gamma^{(1)}=0.8\), respectively. Solid curves show the covertness-optimized Fock-diagonal probe using heralded single-photon states; dotted curves show the sparse two-mode squeezed vacuum (TMSV) quantum illumination (QI). Different colors correspond to different background photon numbers \(\bar n_B\). Yellow markers denote the performance crossover between the covertness-optimized Fock-diagonal scheme and sparse TMSV QI.}
\label{fig:plot_sensing_comp}
\end{figure*}

\label{sec:covert_regime}

Although sparse coding using $\hat{\rho}_A$ derived in Theorem \ref{thm:adjacent_support} optimizes covertness, other states are more suited for communication and sensing.
The optimal choice depends on the mean signal photon number $\bar{n}_S$ as well as the channel parameters.
Thus, Fig.~\ref{fig:plot_comms_comp} plots maximum number $L(\bar{n}_S)$ of covert bits transmissible per $\sqrt{\text{mode}}$ (i.e., covert capacity) for:
\begin{itemize}
    \item on-off keying using the mixture of vacuum and single-photon state, optimized for covertness, and
    \item a circularly-symmetric Gaussian coherent-state code that achieves the Holevo capacity of the bosonic channel.
\end{itemize}
We derive $L(\bar{n}_S)$ for both schemes in Appendix \ref{app:covcomm}.

Covert sensing performance is a function of Chernoff coefficient $\Phi(\bar{n}_S)$  
\cite{audenaert07QCB}. 
We derive these in Appendix \ref{app:covsensing}  and plot the results in Fig.~\ref{fig:plot_sensing_comp} for:
\begin{itemize}
    \item covertness-optimized single-photon probes, and
    \item quantum illumination using two-mode squeezed vacuum. 
\end{itemize}
Both Figs.~\ref{fig:plot_comms_comp} and \ref{fig:plot_sensing_comp} feature crossovers that illustrate the central tradeoff in sparse covert signaling: at small $\bar n_S$, schemes optimized for communication or sensing performance are superior, whereas at larger $\bar n_S$, the improvement in the QRE coefficient achieved by the covertness-optimal input dominates.

\section{Conclusion} \label{sec:conclusion} 
We characterized the optimal covert sparse signaling input state for lossy thermal-noise bosonic channels.
Surprisingly, it is a mixture of at most two consecutive photon-number states. We also examined the trade-offs between these states and those better suited for communication and sensing.
In the near future, we will evaluate additional schemes for covert target detection and extend our results to other sensing problems, e.g., estimating target transmittance $\gamma$ \cite{gong22losssensing}. 
We will also examine the implications of our work for the secret-key size requirement in covert communication. 

\section*{Acknowledgment}
The authors acknowledge helpful discussions with Saikat Guha, Mehrdad Tahmasbi,  Matthieu R.~Bloch, and Uzi Pereg.
This material is based upon work supported by the National Science Foundation under Grants No.~2045530 and 2107265.

\clearpage

\appendices

\section{Proof of Lemma~\ref{lemma:willie_output}}\label{ap:Willie_output}
The input state $\hat\rho_A=\sum_{i=0}^{r}\rho_i\ket{i}\bra{i}$
is invariant under phase rotations. Since the lossy thermal-noise bosonic
channel is phase covariant, its output at Willie is also invariant under
phase rotations. Thus Willie’s
output state is diagonal in the Fock basis. It remains to identify the diagonal probabilities $p_\ell$. For a
number-state input $\ket{i}\bra{i}$, the phase-space description of the thermal-loss channel yields a radial characteristic
function proportional to
$
    e^{-\left(\frac12+\eta\bar n_B\right)x}
    L_i\left((1-\eta)x\right),
    \quad x\geq 0,
$
where $L_i(\cdot)$ is the Laguerre polynomial of degree $i$. Using
the Fock-basis inversion formula for diagonal matrix elements \cite[Ch.~7]{guha04mastersthesis},
the conditional photon-number probability at Willie is
$
    p_{\ell|i}
    =
    \int_{0}^{\infty}
    e^{-(1+\eta\bar n_B)x}
    L_\ell(x)L_i\left((1-\eta)x\right)\dif x .
$
Since the input is the mixture $\sum_{i=0}^{r}\rho_i\ket{i}\bra{i}$,
linearity yields \eqref{eq:pl_integral_lemma}.
Finally, we expand both Laguerre polynomials via
$
    L_i(z)=\sum_{t=0}^{i}(-1)^t\binom{i}{t}\frac{z^t}{t!}.
$
and use
$
    \int_0^\infty e^{-a x}x^{t+u}\dif x
    =
    \frac{(t+u)!}{a^{t+u+1}}
$
with $a=1+\eta\bar n_B$ to obtain the double-summation expression for $p_\ell$ in Lemma \ref{lemma:willie_output}.

\section{Proof of Lemma~\ref{lemma:sparse_qre_expansion}}
\label{ap:proof_QRE_expansion}

Denote
\begin{align}
\hat{\rho}_W
=
\sum_{\ell=0}^{\infty}
p_{\alpha,\ell}\ket{\ell}\bra{\ell},
\end{align}
where,
\begin{align}
p_{\alpha,\ell}
=
(1-\alpha)\lambda_\ell+\alpha p_\ell
=
\lambda_\ell+\alpha(p_\ell-\lambda_\ell),    
\end{align}
$p_\ell$ is in \eqref{eq:pl_integral_lemma} and $\lambda_{\ell}= \frac{(\eta\bar n_\text{B})^\ell} {(1+\eta\bar n_\text{B})^{\ell+1}}$ is the $\ell^{\text{th}}$ diagonal element of $\hat{\rho}_\text{th}\left(\eta\bar n_\text{B}\right)$.
Also, define,
\begin{align}
h_\ell=p_\ell-\lambda_\ell,
~
r_\ell=\frac{h_\ell}{\lambda_\ell}.
\end{align}
Then
\begin{align}
p_{\alpha,\ell}
=
\lambda_\ell(1+\alpha r_\ell).
\end{align}
Since both \(\{\lambda_\ell\}\) and \(\{p_\ell\}\) are probability distributions,
\begin{align}    
\sum_{\ell=0}^{\infty}h_\ell
 =
\sum_{\ell=0}^{\infty}p_\ell
 -
\sum_{\ell=0}^{\infty}\lambda_\ell
=0.
\end{align}
Equivalently,
\begin{align}
\label{eq:lambda_r_zero}
\sum_{\ell=0}^{\infty}\lambda_\ell r_\ell=0.
\end{align}

Now, $ \hat\rho_{\text{th}}(\eta\bar n_B)$ and $\hat{\sigma}_W
=
\sum_{\ell=0}^{\infty}
 p_\ell \ket{\ell}\bra{\ell}$
are both diagonal in the Fock basis. Therefore, the QRE
reduces to the classical relative entropy, 
\begin{align}
D(\hat{\rho}_W\|\hat{\rho}_\text{th}\left(\eta\bar n_\text{B}\right))
&=
\sum_{\ell=0}^{\infty}
p_{\alpha,\ell}
\log\frac{p_{\alpha,\ell}}{\lambda_\ell}               \\
&=
\sum_{\ell=0}^{\infty}
\lambda_\ell(1+\alpha r_\ell)\log(1+\alpha r_\ell).    
\end{align}
Using \eqref{eq:lambda_r_zero}, we subtract the vanishing first-order term:
\begin{align}
D(\hat{\rho}_W\|\hat{\rho}_\text{th}\left(\eta\bar n_\text{B}\right))
&=
\sum_{\ell=0}^{\infty}
\lambda_\ell
\left[
(1+\alpha r_\ell)\log(1+\alpha r_\ell)
-\alpha r_\ell
\right].    \label{eq:cre}
\end{align}
We now justify differentiating this series term by term at $\alpha=0$. First, note that  $r_\ell$ grows at most polynomially,
whereas $\lambda_\ell$ decays geometrically in $\ell$. Hence
\begin{align}
\sum_{\ell=0}^{\infty}\lambda_\ell |r_\ell|<\infty,
\quad
\sum_{\ell=0}^{\infty}\lambda_\ell r_\ell^2<\infty.
\end{align}
Fix
$0<\alpha_0<1$. For $0\le\alpha\le\alpha_0$, since $p_\ell\ge0$, we have
\begin{align}
r_\ell=\frac{p_\ell-\lambda_\ell}{\lambda_\ell}\ge -1.    
\end{align}
Hence
$
1+\alpha r_\ell\ge 1-\alpha_0
$
when $r_\ell<0$, while
$    
1+\alpha r_\ell\ge 1
$
when $r_\ell\ge0$. Therefore,
\begin{align}    
0\le
\frac{\lambda_\ell r_\ell^2}{1+\alpha r_\ell}
\le
\frac{\lambda_\ell r_\ell^2}{1-\alpha_0}.
\end{align}
The right-hand side is summable. Therefore, dominated convergence theorem justifies exchanging the first and second derivatives with respect to $\alpha$ and the infinite sum in \eqref{eq:cre}:
\begin{align}
\left.
\frac{\dif}{\dif\alpha}
D\left(
\hat{\rho}_W
\middle\|
\hat\rho_{\mathrm{th}}(\eta\bar n_B)
\right)
\right|_{\alpha=0}
&=
\sum_{\ell=0}^{\infty}
(p_\ell-\lambda_\ell)
=
0,
\\
\left.
\frac{\dif^2}{\dif\alpha^2}
D\left(
\hat{\rho}_W
\middle\|
\hat\rho_{\mathrm{th}}(\eta\bar n_B)
\right)
\right|_{\alpha=0}
&=
\sum_{\ell=0}^{\infty}
\frac{(p_\ell-\lambda_\ell)^2}{\lambda_\ell}
=
C_W .\label{eq:C_W}
\end{align}
Thus, the Taylor-series expansion yields \eqref{eq:D_expansion} with $C_W$ in \eqref{eq:C_W}.

\section{Proof of Lemma~\ref{lemma:dia_perturb_coef}}
\label{ap:proof_orthogonal_C}

We work under the assumptions of Lemmas~\ref{lemma:willie_output} and \ref{lemma:sparse_qre_expansion}, and use Definition~\ref{def:meixner}. Thus, the input state's basis is truncated per \eqref{eq:truncated_input}
and Willie's output is diagonal,
\begin{align}
\hat{\sigma}_W=\sum_{\ell=0}^{\infty} p_\ell \ket{\ell}\bra{\ell}.
\end{align}
By Lemma~\ref{lemma:sparse_qre_expansion},
\begin{align}
C_W=\sum_{\ell=0}^{\infty}\frac{(p_\ell-\lambda_\ell)^2}{\lambda_\ell},
\label{eq:appC_start}
\end{align}
where
\begin{align}
\lambda_\ell=\frac{(\eta \bar n_B)^\ell}{(1+\eta \bar n_B)^{\ell+1}}.
\label{eq:app_lambda}
\end{align}

From Lemma~\ref{lemma:willie_output}, the output transition probabilities admit the integral
representation
\begin{align}
\mathcal{I}_{\ell,i}
=
\int_{0}^{\infty}
e^{-(1+\eta \bar n_B)x}L_\ell(x)L_i\left((1-\eta)x\right)\dif x,
\label{eq:app_Ili}
\end{align}
so that
\begin{align}
p_\ell=\sum_{i=0}^{r}\rho_i \mathcal{I}_{\ell,i},
\qquad
\lambda_\ell=\mathcal{I}_{\ell,0}.
\label{eq:app_p_lambda_kernel}
\end{align}
Using \(\sum_{i=0}^{r}\rho_i=1\), we rewrite
\begin{align}
p_\ell-\lambda_\ell
=
\sum_{i=1}^{r}\rho_i\left(\mathcal{I}_{\ell,i}-\mathcal{I}_{\ell,0}\right).
\label{eq:app_diff}
\end{align}
Substituting \eqref{eq:app_diff} into \eqref{eq:appC_start} yields
\begin{align}
C_W=
\sum_{\ell=0}^{\infty}
\frac{1}{\lambda_\ell}
\left(
\sum_{i=1}^{r}\rho_i\left(\mathcal{I}_{\ell,i}-\mathcal{I}_{\ell,0}\right)
\right)^2.
\label{eq:appC_quad}
\end{align}
Next, define
\begin{align}
R_i(\ell)=\frac{\mathcal{I}_{\ell,i}}{\lambda_\ell},
\qquad
S(\ell)=\frac{p_\ell}{\lambda_\ell}
=
\sum_{i=0}^{r}\rho_i R_i(\ell).
\label{eq:app_r_S}
\end{align}
Then \eqref{eq:appC_start} becomes
\begin{align}
C_W=\sum_{\ell=0}^{\infty}\lambda_\ell\left(S(\ell)-1\right)^2.
\label{eq:appC_S}
\end{align}

We now identify the orthogonal basis adapted to $\lambda_\ell$.
Let
\begin{align}
\zeta=\frac{\eta \bar n_B}{1+\eta \bar n_B},
\qquad
\theta=\frac{1-\eta}{1+\eta \bar n_B}.
\label{eq:app_q_theta}
\end{align}
Since
\begin{align}
\lambda_\ell=\frac{(\eta \bar n_B)^\ell}{(1+\eta \bar n_B)^{\ell+1}}
=(1-\zeta)\zeta^\ell,
\label{eq:app_lambda_geo}
\end{align}
the sequence $\{\lambda_\ell\}_{\ell\ge 0}$ is geometric. We, therefore, use Meixner
polynomials $M_y(\ell;\zeta)$, whose generating function is
\begin{align}
\sum_{y=0}^{\infty} M_y(\ell;\zeta) \nu^y
=
\left(1-\frac{\nu}{\zeta}\right)^\ell (1-\nu)^{-\ell-1},
\qquad |\nu|<1,
\label{eq:app_meixner_gf}
\end{align}
and whose orthogonality relation under $\lambda_\ell$ is
\begin{align}
\sum_{\ell=0}^{\infty}\lambda_\ell M_y(\ell;\zeta)M_{y'}(\ell;\zeta)
=
\zeta^{-y}\delta_{y,y'}.
\label{eq:app_meixner_orth}
\end{align}

To expand \(R_i(\ell)\) in this basis, consider the generating function
\begin{align}
G_\ell(s)=\sum_{i=0}^{\infty}\mathcal{I}_{\ell,i}s^i.
\label{eq:app_G}
\end{align}
Using the Laguerre polynomial generating function
\begin{align}
\sum_{i=0}^{\infty}L_i\left((1-\eta)x\right)s^i
=
\frac{1}{1-s}
\exp\left(
-\frac{(1-\eta)xs}{1-s}
\right),
\label{eq:app_laguerre_gf}
\end{align}
together with \eqref{eq:app_Ili}, we obtain
\begin{align}
G_\ell(s)
=
\frac{1}{1-s}
\int_{0}^{\infty}
\exp\left(
-\left[1+\eta \bar n_B+\frac{(1-\eta)s}{1-s}\right]x
\right)
L_\ell(x)\dif x.
\label{eq:app_G2}
\end{align}
Applying the Laplace-transform identity for Laguerre polynomials,
\begin{align}
\int_{0}^{\infty}e^{-ax}L_\ell(x)dx=\frac{(a-1)^\ell}{a^{\ell+1}},
\qquad a>1,
\label{eq:app_laguerre_laplace}
\end{align}
yields
\begin{align}
\sum_{i=0}^{\infty} R_i(\ell)s^i&=
\frac{G_\ell(s)}{\lambda_\ell}\\
&=
\left(1+(\theta/\zeta-1)s\right)^\ell
\left(1+(\theta-1)s\right)^{-\ell-1}.
\label{eq:app_ri_gf}
\end{align}
On the other hand, using \eqref{eq:app_meixner_gf} with $\nu=-\frac{\theta~s}{1-s}$, together with the binomial-series identity
\begin{align}
\sum_{i=y}^{\infty}\binom{i}{y}s^i=\frac{s^y}{(1-s)^{y+1}},
\qquad y=0,1,2,\dots,
\label{eq:app_binomial_series}
\end{align}
we obtain
\begin{align}
&\sum_{i=0}^{\infty}
\left(
\sum_{y=0}^{i}\binom{i}{y}(-\theta)^y M_y(\ell;\zeta)
\right)s^i \notag\\
&= 
\left(1+(\theta/\zeta-1)s\right)^\ell
\left(1+(\theta-1)s\right)^{-\ell-1}.
\label{eq:app_compare_gf}
\end{align}
Comparing \eqref{eq:app_ri_gf} and \eqref{eq:app_compare_gf} shows that, for every \(i\ge 0\),
\begin{align}
R_i(\ell)=\sum_{y=0}^{i}\binom{i}{y}(-\theta)^y M_y(\ell;\zeta).
\label{eq:app_ri_expansion}
\end{align}
Since \(R_0(\ell)=1\), it follows from \eqref{eq:app_r_S} and \eqref{eq:app_ri_expansion} that
\begin{align}
S(\ell)-1
&=
\sum_{i=1}^{r}\rho_i\left(R_i(\ell)-1\right) \nonumber\\
&=
\sum_{i=1}^{r}\rho_i
\sum_{y=1}^{i}\binom{i}{y}(-\theta)^y M_y(\ell;\zeta) \nonumber\\
&=
\sum_{y=1}^{r}
\left[
(-\theta)^y
\sum_{i=y}^{j}\rho_i\binom{i}{y}
\right]
M_y(\ell;\zeta).
\label{eq:app_Sminus1}
\end{align}
Define
\begin{align}
\mu_y=\sum_{i=y}^{r}\rho_i\binom{i}{y},
\qquad y=1,2,\dots,r.
\label{eq:app_mu}
\end{align}
Then \eqref{eq:app_Sminus1} becomes
\begin{align}
S(\ell)-1=\sum_{y=1}^{r} b_y M_y(\ell;\zeta),
\quad
b_y=(-\theta)^y\mu_y.
\label{eq:app_by}
\end{align}
Finally, substituting \eqref{eq:app_by} into \eqref{eq:appC_S} gives
\begin{align}
C_W
&=
\sum_{\ell=0}^{\infty}\lambda_\ell
\left(
\sum_{y=1}^{r}b_y M_y(\ell;\zeta)
\right)^2 \nonumber\\
&=
\sum_{y=1}^{j}\sum_{y'=1}^{r}
b_y b_{y'}
\sum_{\ell=0}^{\infty}\lambda_\ell M_y(\ell;\zeta)M_{y'}(\ell;\zeta).
\label{eq:appC_expand}
\end{align}
Applying the orthogonality relation \eqref{eq:app_meixner_orth}, all cross terms vanish and we obtain
\begin{align}
C_W=\sum_{y=1}^{r} \zeta^{-y} b_y^2
=
\sum_{y=1}^{r} \zeta^{-y}\theta^{2y}\mu_y^2,
\quad
\mu_y=\sum_{i= y}^r\rho_i\binom{i}{y}.
\label{eq:appC_final}
\end{align}
which is precisely the claimed orthogonal representation. 

\section{Proof of Lemma~\ref{lemma:two_p_lb}}
\label{ap:two_p_lb}
Let
$
\mathbb{E}[\mathrm{I}]=\kappa+\beta
$,
where
$
\kappa=\lfloor \mathbb{E}[\mathrm{I}]\rfloor
$
and
$
\beta\in[0,1)
$.
Define
\begin{align}
f=\varphi(\tau+1)-\varphi(\tau).
\end{align}
Since $\varphi$ is discrete convex, the increments
$
\varphi(i+1)-\varphi(i)
$
are nondecreasing in $i$. We claim that
\begin{align}
\varphi(i)\ge \varphi(\kappa)+f(i-\kappa),
\qquad i\in \mathbb{Z}_{\ge 0}.
\label{eq:supporting_line_discrete}
\end{align}

For $i\ge \kappa$,
\begin{align}
\varphi(i)-\varphi(\kappa)
=
\sum_{\kappa^\prime=\kappa}^{i-1}\left(\varphi(\kappa^\prime+1)-\varphi(\kappa^\prime)\right)
\ge
\sum_{\kappa^\prime=\kappa}^{i-1} f
=
f(i-\kappa),
\end{align}
because each increment in the sum is at least $f$. 

For $i\le \kappa$,
\begin{align}
\varphi(\kappa)-\varphi(i)
=
\sum_{\kappa^\prime=i}^{\kappa-1}\left(\varphi(\kappa^\prime+1)-\varphi(\kappa^\prime)\right)
\le
\sum_{\kappa^\prime=i}^{\kappa-1} f
=
f(\kappa-i),
\end{align}
because each increment in the sum is at most $f$. Rearranging again yield
\eqref{eq:supporting_line_discrete}. Thus \eqref{eq:supporting_line_discrete} holds for all
$i\ge 0$.

Taking expectation with respect to $\mathrm{I}$ yields
\begin{align}
\mathbb{E}[\varphi(\mathrm{I})]
\ge
\mathbb{E}\left[\varphi(\kappa)+f(\mathrm{I}-\kappa)\right]
&=
\varphi(\kappa)+f\left(\mathbb{E}[\mathrm{I}]-\kappa\right)\notag\\
&=
\varphi(\kappa)+f\beta.
\end{align}
Substituting \(f=\varphi(\kappa+1)-\varphi(\kappa)\), we obtain
\begin{align}
\mathbb{E}[\varphi(\mathrm{I})]
\ge
(1-\beta)\varphi(\kappa)+\beta \varphi(\kappa+1),
\end{align}
which proves the inequality.

Finally, if
$
\Pr[\mathrm{I}=\kappa]=1-\beta
$
and
$
\Pr[\mathrm{I}=\kappa+1]=\beta
$,
then \(\mathbb{E}[\mathrm{I}]=\kappa+\beta\) and
\begin{align}
\mathbb{E}[\varphi(\mathrm{I})]
=
(1-\beta)\varphi(\kappa)+\beta \varphi(\kappa+1),
\end{align}
which proves the equality and completes the proof.

\section{Covert Communication Performance Analysis}
\label{app:covcomm}

\subsection{On-Off Keying with a Single-Photon State}
Suppose Alice uses on-off keying, where $\hat{\rho}_A$ from Section \ref{subsec:low_brightness_sparse_optimum} and vacuum $\ket{0}\bra{0}$ modulate the ``on'' and ``off'' symbols, respectively. 
Alice and Bob use a blocklength-$n$ code such that $p(\text{``on''})=\alpha$ and  $p(\text{``off''})=1-\alpha$, where setting $\alpha=\sqrt{\frac{\delta_{\text{QRE}}}{2nC_W}}$ with $C_W$ from Section \ref{subsec:low_brightness_sparse_optimum} maintains the covertness constraint in \eqref{eq:cov_req_comms}.
Furthermore, suppose Bob employs a photon-number-resolving (PNR) measurement independently on each of his $n$ output modes. This induces a classical discrete memoryless channel between Alice and Bob, with transition probabilities 
\begin{align}
\label{eq:transi_prob_0}
p\left(\ell\middle|\text{``off''}\right)&=\frac{\left( \left(1-\eta\right) \bar n_B \right)^\ell}{(1+ (1-\eta) \bar n_B)^{\ell+1}}\\
p\left(\ell\middle|\text{``on''}\right)&=
(1-\bar n_S)
\frac{\big((1-\eta)\bar n_B\big)^\ell}
{\big(1+(1-\eta)\bar n_B\big)^{\ell+1}}
\nonumber\\
&\phantom{=}+\bar n_S\sum_{u=0}^{\ell}
(-1)^{u}
\binom{\ell}{u}
\frac{1}{\left(1+(1-\eta)\bar n_B\right)^{u+1}}\nonumber\\
&\phantom{=}+\bar n_S
\sum_{u=0}^{\ell}
(-1)^{u+1}
\binom{\ell}{u} 
\frac{\eta(u+1)}{\left(1+(1-\eta)\bar n_B\right)^{u+2}},\label{eq:transi_prob_1} 
\end{align}
where $\ell\in\{0,1,2,\ldots\}$ is Bob's PNR measurement output. 
We have \eqref{eq:transi_prob_0} since
Bob's output is a thermal state with a geometric photon number distribution when Alice transmits vacuum.
Evaluating $p_\ell$ in Lemma \ref{lemma:willie_output} with the substitution $\eta\to1-\eta$ to obtain Bob's rather than Willie's photon count yields \eqref{eq:transi_prob_1}.
The covert capacity follows from \cite[Th.~2]{bloch15covert}:
\begin{align}    
L(\bar n_S)&=(1-\xi)
\sqrt{\frac{2\delta_{\text{QRE}}}{C_W}}D\left(p\left(\ell\middle|\text{``on''}\right)\middle\|p\left(\ell\middle|\text{``off''}\right)\right),
\end{align}
where $0<\xi<1$ is an auxiliary constant in the proof of \cite[Th.~2]{bloch15covert} which we set to unity for simplicity of exposition, and $D\left(p(x)\middle\|q(x)\right)=\sum_x p(x)\log\frac{p(x)}{q(x)}$ is the classical relative entropy with summation over the domain of density $p(x)$.

\subsection{Holevo-Capacity-Achieving Modulation}
\label{app:holevocapmod}
Now suppose that Alice and Bob use a code that achieves the Holevo capacity on the subset of modes that are selected for transmission in Section \ref{subsec: Signaling Model}. 
Thus, Alice transmits a coherent state whose amplitude is drawn from a zero-mean circularly-symmetric complex Gaussian distribution with mean photon number $\bar n_S$, i.e., $\hat{\rho}_A=\ket{\epsilon}\bra{\epsilon}$, where
$\epsilon \sim \mathcal{CN}(0,\bar n_S)$ \cite{giovannetti_ultimate_2014}.
The Holevo capacity per transmitted mode is then
\begin{align}
\chi_\text{Hol}(\bar n_S)
=
g\left((1-\eta)\bar n_B+\eta\bar n_S\right)
-
g\left((1-\eta)\bar n_B\right),
\label{eq:chi2b_def}
\end{align}
where $g(x)=(x+1)\log(x+1)-x\log x$.

If Alice and Bob keep their code secret from Willie, then, in each mode, he receives a mixture of thermal states $\hat{\rho}_W=(1-\alpha)\hat{\rho}_{\text{th}}\left(\eta\bar{n}_B\right)+\alpha\hat{\rho}_{\text{th}}\left(\eta\bar{n}_B+(1-\eta)\bar{n}_S\right)$.
This is because averaging over Gaussian-modulated states yields a thermal state: $\hat{\rho}_{\text{th}}\left(\bar{n}_S\right)=\frac{1}{\pi \bar n_S}\int_{\mathbb C}
e^{-|\epsilon|^2/\bar n_S}
\ket{\epsilon}\bra{\epsilon}\dif^2\epsilon$.
We use the diagonality of $\hat{\rho}_W$ in the Fock basis to expand the QRE $D\left(\hat{\rho}_W\middle\|\hat{\rho}_{\text{th}}\left(\eta\bar{n}_\text{B}\right)\right)
=
\frac{\alpha^2}{2} C_G+o(\alpha^2)$,
where 
\begin{align}
C_G&=
\frac{(1-\eta)^{2} \bar n_S^{2}}
{\eta \bar{n}_B \left(1+\eta \bar{n}_B\right)-(1-\eta)^{2} \bar n_S^{2}}.\label{eq:C_G}
\end{align}
Note that $C_G>C_W$.
The resulting covert capacity is
$L(\bar n_S)=\sqrt{\frac{2\delta_{\text{QRE}}}{C_{G}}}\chi_\text{Hol}(\bar n_S)$.

Finally, note that here $L(\bar{n}_S)$ is an expected value with respect to the random sequence of coin flips that determines the modes used for transmission. There is a small chance that Alice and Bob select a sequence with too few or (too many) modes.  We can avoid this by expurgating such bad sequences by applying the argument outlined in the construction of \cite[Lemma 3]{anderson25coventgenbosonic} and \cite[App.~VI]{anderson25coventgenbosonic}, which is, in turn, adapted from \cite{tahmasbi21signalingcovert}.
Indeed, the original expurgation argument in \cite{tahmasbi21signalingcovert} addresses the covert sensing setting.
Note that this technique also applies to the analysis in the following appendix.

\section{Covert Sensing Performance Analysis}
\label{app:covsensing}
\subsection{Covertness-Optimized Single-Photon Probes}
\label{app:sensignsinglephoton}
Suppose Alice uses $\hat{\rho}_S=\hat{\rho}_A$ from Section \ref{subsec:low_brightness_sparse_optimum} to probe the target, without using reference modes $I$.
Substitutions $\tau_W\to1-\eta$,  $\gamma^{(h)}\bar n_S\to\bar n_S$ and $\bar n_W\to\bar{n}_B$ in $C_W$ from that section yield the expansion coefficient in \eqref{eq:D_expansion}, $C_{W_h}=\frac{\left(\gamma^{(h)}\tau_W\bar n_S\right)^2}
{(1-\tau_W)\bar n_W\left[1+(1-\tau_W)\bar n_W\right]}$, 
with $h\in\{0,1\}$ indicating whether the target is absent or present.  
Setting $\alpha=\min_{h\in\{0,1\}}\sqrt{\frac{\delta_{\text{QRE}}}{2nC_{W_h}}}$  maintains the covertness constraint in \eqref{eq:cov_req_sensing}.
We assume that Alice uses a PNR measurement on her returned probe.
The likelihood of the target-presence hypothesis $h$ is a function of the detected photon number $\ell$:
\begin{align}
p_\ell(h)&=(1-\bar n_S)
\frac{\left((1-\tau_A)\bar n_A\right)^\ell}
{\left(1+(1-\tau_A)\bar n_A\right)^{\ell+1}}
\nonumber\\
&\phantom{=}+\bar n_S\sum_{u=0}^{\ell}(-1)^u\binom{\ell}{u}\frac{1}{\left(1+(1-\tau_A)\bar n_A\right)^{u+1}}\nonumber\\
&\phantom{=}+\bar n_S\sum_{u=0}^{\ell}(-1)^{u+1}
\binom{\ell}{u}\frac{(1-\gamma^{(h)})\tau_A(u+1)}
{\left(1+(1-\tau_A)\bar n_A\right)^{u+2}},
\end{align}
We derive it by substituting $(1-\tau_A)\bar n_A\to (1-\eta)\bar n_B$ and $(1-\gamma^{(h)})\tau_A\to\eta$ in \eqref{eq:transi_prob_1}.
An upper bound for Alice's detection error probability is:
 \begin{align}
 P_e^{\text{(s)}}&\leq \frac{1}{2}\exp\left[-n\alpha \Phi(\bar{n}_S)\right],\label{eq:chernoff_bound}
 \end{align}
where $\Phi(\bar{n}_S)=-\log\min_{0\le s\le 1}\sum_{\ell=0}^{\infty}\left(p_{\ell}\left(0\right)\right)^s\left(p_{\ell}\left(1\right)\right)^{1-s}$ is the classical Chernoff exponent \cite[Ch.~2.4]{vantrees01part1}. 
Thus, $\lim_{n\to\infty}-\frac{1}{\sqrt{n}}\log P_e^{\text{(s)}}=\min_{h\in\{0,1\}}\sqrt{\frac{\delta_{\text{QRE}}}{2C_{W_h}}}\Phi(\bar{n}_S)$ quantifies the covert sensing capability that we plot in Fig.~\ref{fig:plot_sensing_comp}.
We note that, although here the probability of error $P_e^{\text{(s)}}$ is, in fact, an expected value with respect to the random sequence of coin flips that determines the modes used for probing.
However, just as in Appendix \ref{app:holevocapmod}, we can expurgate the bad sequences per \cite{tahmasbi21signalingcovert}, \cite[Lemma 3 and App.~VI]{anderson25coventgenbosonic}.

\subsection{Quantum Illumination Using Two-Mode Squeezed Vacuum}
Now suppose Alice uses quantum illumination with two-mode squeezed vacuum (TMSV) states \cite{tan08qigaussianstates} $\hat{\rho}_{IS}=\ket{\psi}\bra{\psi}_{IS}$, where $\ket{\psi}_{IS}=\sum_{i=0}^\infty\sqrt{\frac{(\bar{n}_S)^i}{(1+\bar{n}_S)^{i+1}}}\ket{i}_I\ket{i}_S$. We consider TMSV state because they optimize the quantum Chernoff exponent among Gaussian inputs, are near-optimal in the low-reflectance regime \cite{bradshaw2021optimal}, and optimize the missed-detection/false-alarm tradeoff among all input states \cite{depalma18mepQI}.
Since Willie cannot access reference modes $I$, when Alice transmits, he receives a thermal state $\hat{\rho}_{\text{th}}\left((1-\tau_W)\bar{n}_W+\gamma^{(h)}\tau_W\bar{n}_S\right)$.
Adapting the approach from Appendix \ref{app:holevocapmod}, we obtain $C_G$ as in \eqref{eq:C_G} with substitutions $\tau_W\to1-\eta$ and $\bar{n}_W\to\bar{n}_B$.
We eliminate $\gamma^{(h)}$ by noting that, since it multiplies $\bar{n}_S$, taking $h=0$ maximizing $C_G$.
Furthermore, \eqref{eq:chernoff_bound} holds with quantum Chernoff exponent $\Phi(\bar n_S)=-\log\min_{0\le s\le 1}\tr\left[\left(\hat{\rho}_{IR}(0)\right)^s\left(\hat{\rho}_{IR}(1)\right)^{1-s}\right]$, where $\hat{\rho}_{IR}(h)$ is the joint reference and probe return state when target is either absent ($h=0$) or present ($h=1$).
As in Appendix \ref{app:sensignsinglephoton}, we quantify the covert sensing capacity in Fig.~\ref{fig:plot_sensing_comp} by $\lim_{n\to\infty}-\frac{1}{\sqrt{n}}\log P_e^{\text{(s)}}=\sqrt{\frac{\delta_{\text{QRE}}}{2C_{G}}}\Phi(\bar{n}_S)$.


\end{document}